\documentclass[11pt]{article}
\usepackage[utf8]{inputenc}
\usepackage{hyperref}
\usepackage{amsmath,amsthm,amssymb,caption,subcaption,cleveref,nicefrac}
\usepackage{fullpage}
\usepackage{bm}
\usepackage[table]{xcolor}
\usepackage{multirow}
\usepackage{arydshln}
\usepackage{csquotes}
\usepackage{anyfontsize}
\usepackage[colorinlistoftodos,prependcaption,textsize=tiny]{todonotes}
\usepackage[shortlabels]{enumitem}

\usepackage{algorithm}
\usepackage{algorithmic}
\usepackage[algo2e,ruled,noend]{algorithm2e}
\usepackage{xspace}

\newtheorem{theorem}{Theorem}

\newtheorem{lemma}[theorem]{Lemma}
\newtheorem{corollary}[theorem]{Corollary}

\newcommand{\cA}{\mathcal{A}}

\newcommand{\eps}{\epsilon}
\newcommand{\N}{\mathbb{N}}

\renewcommand{\phi}{\varphi}

\renewcommand{\setminus}{\smallsetminus}

\newcommand{\tu}{\tilde{u}}

\newcommand{\cC}{\mathcal{C}}

\newcommand{\argmax}{\mathrm{argmax}}
\newcommand{\cB}{\mathcal{B}}

\newcommand{\cV}{\mathcal{V}}
\newcommand{\Gin}{G^{\mathrm{inc}}}
\newcommand{\Ein}{E^{\mathrm{inc}}}
\newcommand{\opt}{\mathrm{OPT}}
\newcommand{\val}{\mathrm{val}}
\newcommand{\cH}{\mathcal{H}}

\newcommand{\negg}{\mathrm{neg}}
\newcommand{\tcV}{\widetilde{\cV}}
\newcommand{\ifin}{i_{\mathrm{fin}}}
\newcommand{\cCdel}{\cC_{\mathrm{del}}}
\newcommand{\cCdist}{\cC_{\mathrm{set}}}
\newcommand{\UNSAT}{\mathrm{UNSAT}}
\newcommand{\SAT}{\mathrm{SAT}}
\newcommand{\topt}{\widetilde{\opt}}

\newcommand{\npdeg}{\mathrm{nndeg}}
\newcommand{\Iden}{\mathrm{Iden}}

\allowdisplaybreaks

\title{Improved FPT Approximation Scheme and Approximate Kernel for Biclique-Free Max $k$-Weight SAT: Greedy Strikes Back}
\author{{Pasin Manurangsi}\\ Google Research, Thailand\\ \texttt{pasin@google.com}}
\date{\today}

\begin{document}

\maketitle

\begin{abstract}
In the \emph{Max $k$-Weight SAT} (aka \emph{Max SAT with Cardinality Constraint}) problem, we are given a CNF formula with $n$ variables and $m$ clauses together with a positive integer $k$. The goal is to find an assignment where at most $k$ variables are set to one that satisfies as many constraints as possible. Recently, Jain et al.~\cite{JKPSSSU23} gave an FPT approximation scheme (FPT-AS) with running time $2^{O\left(\left(dk/\eps\right)^d\right)} \cdot (n + m)^{O(1)}$ for Max $k$-Weight SAT when the incidence graph is $K_{d,d}$-free. They asked whether a polynomial-size approximate kernel exists. In this work, we answer this question positively by giving an $(1 - \eps)$-approximate kernel with $\left(\frac{d k}{\eps}\right)^{O(d)}$ variables. This also implies an improved FPT-AS with running time $(dk/\eps)^{O(dk)} \cdot (n + m)^{O(1)}$. Our approximate kernel is based mainly on a couple of greedy strategies together with a sunflower lemma-style reduction rule.
\end{abstract}

\section{Introduction}
In the \emph{Max $k$-Weight SAT} problem (aka the \emph{Max SAT with Cardinality Constraint} problem), we are given a CNF formula $\Phi= (\cV, \cC)$, where $\cV$ is the set of $n$ variables and $\cC$ denotes the multiset\footnote{Multiset is more convenient for our algorithms. We provide a discussion on multiset-vs-set in \Cref{sec:open}.} of $m$ clauses. The weight of an assignment is the number of variables set to true. The goal here is to output an assignment of weight at most $k$ that satisfies the maximum number of constraints. 

Max $k$-Weight SAT and its many special cases have long been studied in the approximation algorithm literature (e.g.~\cite{Feige98,AgeevS99,AgeevS04,FeigeL01,Sviridenko01,BlaserM02,Hofmeister03,RaghavendraT12,AustrinBG16,Man19,ZadehBGNSS22}). Sviridenko~\cite{Sviridenko01} gave a polynomial-time $\left(1 - \frac{1}{e}\right)$-approximation algorithm for the problem. Since Feige~\cite{Feige98} had earlier proved that $\left(1 - \frac{1}{e} + o(1)\right)$-approximation is NP-hard, this settles the polynomial-time approximability of the problem. In fact, Feige proved the hardness of approximation even for the special case where the formula is monotone (i.e. all literals are positive), which is often referred to as the \emph{Max $k$-Coverage} problem. The simple greedy algorithm for Max $k$-Coverage, which also yields the tight $\left(1 - \frac{1}{e}\right)$-approximation, have been known since the 70's~\cite{NemhauserWF78}. Special cases of Max $k$-Coverage are also studied. For example, when we assume that each clause contains $p$ literals, this corresponds to the so-called \emph{Max $k$-Vertex Cover in $p$-Uniform Hypergraph} which has been studied in \cite{AgeevS99,AgeevS04,FeigeL01,Man19,AustrinS19}. Even this special case remains an active area of research to this day; recently, \cite{AustrinS19} provides tight approximation ratio for the $p = 2$ case (which we will refer to as \emph{Max $k$-Vertex Cover}) but the case $p > 2$ remains open.

Max $k$-Weight SAT has also been extensively studied from the perspective of parameterized complexity. Here $k$ is the parameter, and we wish to find a fixed-parameter tractable (FPT)\footnote{For more background on FPT, see~\cite{CFKLMPPS15}.} algorithms, i.e. one that runs in time $f(k) \cdot n^{O(1)}$ for some computable function $f$. In the seminal work of Downey and Fellows~\cite{DowneyF95}, it is already shown that the decision versions\footnote{The goal here is to decide whether all clauses can be satisfied.} of Max $k$-Weight SAT and Max $k$-Coverage are complete for the class W[2], ruling out the existence of exact FPT algorithms for these problems. Later works show that even achieving $\left(1 - \frac{1}{e} - \eps\right)$-approximation (for any constant $\eps > 0$) in FPT time is impossible assuming Gap-ETH\footnote{Gap Exponential Time Hypothesis (Gap-ETH)~\cite{Dinur16,ManurangsiR17} states that there is no $2^{o(n)}$-time algorithm that distinguish between a satisfiable 3-CNF formula and one which is not even $(1 - \eps)$-satisfiable for some constant $\eps > 0$.}~\cite{CGKLL19,Man20}. Given such strong lower bounds, positive results for this problem have focused on special cases. The first positive result of this kind is due to Marx~\cite{Marx08} who obtained the first FPT approximation scheme (FPT-AS) for Max $k$-Vertex Cover, one that achieves $(1 - \eps)$-approximation in time  $(k/\eps)^{O(k^3/\eps)} \cdot (n + m)^{O(1)}$ for any $\eps > 0$. The runnning time was later improved to $(1/\eps)^{O(k)} \cdot (n + m)^{O(1)}$ by two independent works~\cite{Man19,SkowronF17}. In fact, Skowron and Faliszewski~\cite{SkowronF17} showed that this technique even works for $p$-uniform hypergraph for any constant $p$. Moreover, \cite{Man19} noted that the running time of this FPT-AS is essentially tight: any $(1/\eps)^{o(k)}$-time FPT-AS would break ETH\footnote{Exponential Time Hypothesis (ETH)~\cite{IP01,IPZ01} states that there is no $2^{o(n)}$-time algorithm that solves 3SAT.}.

A very recent work of Jain et al.~\cite{JKPSSSU23} observes that these algorithms rely on certain ``sparsity'' structures of the \emph{incidence graph} of $\Phi$. Recall that the \emph{incidence graph} (aka \emph{clause-variable graph}) of $\Phi$, denoted by $\Gin_{\Phi}$, is the bipartite graph $\Gin_{\Phi} = (\cV, \cC, \Ein_{\Phi})$ such $(v, C) \in \Ein_{\Phi}$ iff $v \in C$ or $\neg v \in C$. For a graph class $\cH$, the \emph{$\cH$ Max $k$-Weight SAT} problem is the problem when we restrict to only instances where $\Gin$ belongs to $\cH$. Note that Max $k$-Vertex Cover on $p$-Uniform Hypergraph belongs to $K_{p+1, 1}$-free Max $k$-Weight SAT, where $K_{a, b}$ denote the complete bipartite graph with $a$ left vertices and $b$ right vertices. With this in mind, Jain et al.~\cite{JKPSSSU23} significantly extends the aforementioned algorithms~\cite{Marx08,Man19,SkowronF17} by giving a FPT-AS for $K_{d, d}$-free Max $k$-Weight SAT with running time $2^{O\left(\left(dk/\eps\right)^d\right)} \cdot (n + m)^{O(1)}$ for any $d \in \N, \eps > 0$. Given that many sparse graph classes are $K_{d,d}$-free for some $d$, this immediately yields FPT-AS for Max $k$-Weight SAT for these graph classes (including bounded treewidth and bounded genus graphs) too\footnote{See Figure 1 of \cite{JKPSSSU23} for more graph classes that are $K_{d,d}$-free.}. Despite the generality of this result, there are still a few remaining open questions. First, is the $(k/\eps)^d$ dependenecy in the exponent of the running time necessary? Second, their technique does not yield a polynomial-size \emph{approximate kernel}, as we will discuss more below.

Kernelization is a central concept in FPT (see e.g.~\cite{kernelization-book}). In the context of FPT approximation algorithms, \cite{LPRS17} define \emph{approximate kernel} as follows. First, we define \emph{$\alpha$-approximate polynomial-time pre-processing algorithm ($\alpha$-APPA)} for a parameterized optimization problem $\Pi$ as a pair of polynomial-time algorithms $\cA, \cB$, called the \emph{reduction algorithm} and the \emph{solution-lifting algorithm} respectively, such that the following holds: (i) Given any instance $(I, k)$ of $\Pi$, $\cA$ outputs an instance $(I', k')$ of $\Pi$, and (ii) given any $\beta$-approximate solution of $(I', k')$, $\cB$ outputs an $\alpha\beta$-approximate solution of $(I, k)$. An \emph{$\alpha$-approximate kernel} is an $\alpha$-APPA such that the output size $|I'| + k'$ is bounded by some computable function of $k$. A fundamental theorem in \cite{LPRS17} is that there is an $\alpha$-approximation FPT algorithm for problem $\Pi$ if and only if it admits an $\alpha$-approximate kernel. For Max $k$-Vertex Cover, \cite{Man19,SkowronF17} actually gave very simple kernel based on greedy strategies: keep only the $O(k/\eps)$ highest degree vertices! Despite this, Jain et al.'s algorithm~\cite{JKPSSSU23} does not yield any explicit kernel. Applying the generic equivalence only gives an approximate kernel whose size is exponential in $k$. As such, they posed an open question whether we can get $k^{O(1)}$-size $(1 - \eps)$-approximate kernel for $K_{d, d}$-free Max $k$-Weight SAT. They highlighted that this is open even for $K_{d, d}$-free Max $k$-Coverage and that this ``seems difficult''.

\subsection{Our Contributions}

Our main contribution is a positive answer to their question: We design an $(1 - \eps)$-approximate kernel for $K_{d, d}$-free Max $k$-Weight SAT whose size is polynomial in $k$. Since our argument is quite flexible, we state the bound below even for $K_{a, b}$-free where $a, b$ may not be equal.

\begin{theorem} \label{thm:main-kernel}
For any $a, b \in \N$ and $\eps \in (0, 1/2)$, there is a parameter-preserving\footnote{We say that an approximate kernel is \emph{parameter-preserving} if we have $k' = k$.} $(1 - \eps)$-approximate kernel for $K_{a,b}$-free Max $k$-Weight SAT with $O\left(\frac{k \log k}{\eps}\right) + a \cdot O(b)^{2b} \cdot \frac{k^b}{\eps^{3b}}$ variables and $(k/\eps)^{O(ab)}$ clauses.
\end{theorem}

In terms of the number of variables, the second term dominates for $b \geq 2$ and we get $O_{a,b}\left(k^b/\eps^{3b}\right)$ variables. For $b = 1$, we get $O_a\left(\frac{k \log k}{\eps} + \frac{k}{\eps^3}\right)$ variables. Up to $O\left(\log k + \frac{1}{\eps^2}\right)$ factor, this latter bound matches the aforementioned approximate kernels for Max $k$-Vertex Cover~\cite{SkowronF17,Man19}. 

Note that any parameter-preserving $(1 - \eps)$-approximate kernel with $n'$ variables allows us to get an $O(n'/k)^k \cdot (n + m)^{O(1)}$-time algorithm: by brute-force trying out all solutions in the reduced instance. (Note that there are only $\binom{n'}{0} + \cdots \binom{n'}{k} \leq O(n'/k)^k$ such solutions.) Plugging this into the above bound, we immediately get the following algorithm:

\begin{corollary} \label{cor:alg}
For any $a, b \in \N$ and $\eps \in (0, 1/2)$, there is an $(1 - \eps)$-approximation algorithm for $K_{a,b}$-free Max $k$-Weight SAT that runs in time 
\begin{itemize}
\item $\left(\frac{\log k + a}{\eps}\right)^{O(k)} \cdot (n + m)^{O(1)}$ if $b = 1$, and,
\item $\left(\frac{a^{1/b} \cdot b k}{\eps}\right)^{O(b k)} \cdot (n + m)^{O(1)}$ if $b > 1$.
\end{itemize}
\end{corollary}

In the case $a = b = d$, the running time of our algorithm is $(dk/\eps)^{O(dk)}$;
this represents an improvement over the running time of $2^{O\left(\left(dk/\eps\right)^d\right)} \cdot (n + m)^{O(1)}$ due to \cite{JKPSSSU23}. (We note that this improvement is only for Max $k$-Weight SAT, as Jain et al.~\cite{JKPSSSU23} already gave FPT-AS with a similar running time for the problem for $a = b = d$.) Furthermore, as mentioned earlier, a running time lower bound of $(1/\eps)^{\Omega(k)}$ holds even for FPT-AS for Max $k$-Vertex Cover (i.e. $a = 3, b = 1$)~\cite{Man19}. Thus, our running time is tight up to the factor of $O_{a,b}(k)$ in the base.

\paragraph{Technical Overview.}
We say that a variable $v$ is \emph{positive} w.r.t. a formula $\Phi$ if the formula does not contain the negative literal $\neg v$. Otherwise, we say that $v$ is \emph{negative} (w.r.t. $\Phi$).

There are three steps in our reduction algorithm: (I) reducing \# of negative variables, (II) reducing \# of positive variables and (III) reducing \# of clauses. 

Step I is based on the following observation: If $n \gg k/\eps$ and the incidence graph for negative literals is bi-regular (i.e. every clause has the same number of negative literals and every negative literals are in the same number of clauses), then \emph{any} solution will satisfy an $(1 - \eps)$ fraction of these clauses. (This is essentially because only an $\eps$ fraction among these literals can be false.) In other words, any solution is an $(1 - \eps)$-approximate solution. More generally, if we can find a subset of clauses satisfying this condition, then we can delete all these clauses in our reduction algorithm. Unfortunately, some formulae can be highly irregular, e.g. each negative literal can have a very different number of occurrences. To deal with this, we define a \emph{normalized negative degree} for each variable. We then iteratively pick variables whose normalized negative degree are above a certain threshold. Once no such variable exists, we stop and delete all clauses that contain at least one of the negative literals that are not chosen. Similar to above, it can be seen that any solution satisfies $1 - \eps$ fraction of the deleted clauses. Meanwhile, by a careful argument, we show that this procedure picks only $O\left(\frac{k \log k}{\eps}\right)$ variables, meaning that only those many vertices have negative literals left.

Step II is actually almost the same as those in previous work~\cite{Man19,SkowronF17}: Just keep a certain number of variables with highest degree (along with those picked in Step I). We show that the sparsity from $K_{a, b}$-freeness is also sufficient for this kernel as long as the optimum is large.
If the optimum is small, then we use a sunflower lemma-based kernel similar to that of (exact) $d$-Set Packing (e.g.~\cite{DellM12})\footnote{We remark that a very recent work of Jain et al.~\cite{JKPSSSU24} gives an (exact) kernel for $K_{d,d}$-free Max $k$-Weight SAT with the desired optimum as the parameter. We could also use their kernel to handle the small-optimum case as well, although the number of variables is slightly worse in their work.}.
For this, we also prove a sunflower lemma for $K_{a, b}$-free graphs that gives improved bounds in a certain regime of parameters, which might be of independent interest.
Finally, for Step III, it should be noted that if at this point $\cC$ is a \emph{set} (instead of \emph{multiset}), then we would have been done because $K_{a, b}$-freeness immediately implies that there are at most $O(b \cdot n^{a})$ distinct clauses. Thus, Step III is essentially a repetition reduction algorithm; applying known techniques from the literature~\cite{CST01} (namely scaling and rounding) immediately yields the claimed result.


\paragraph{On Independent Work of Inamdar et al.~\cite{IJLSSU24}.} Independently of our work, Inamdar et al.~\cite{IJLSSU24} has obtained a set of related results. Compared to our work, the most relevant result from~\cite{IJLSSU24} is a \emph{randomized} FPT-AS for $K_{d,d}$-free Max $k$-Weight SAT that runs in time $(dk/\eps)^{O(dk)} \cdot (n + m)^{O(1)}$. This is exactly the same running time as ours, but their algorithm is randomized whereas ours is deterministic. Their result is based on an elegant randomized assignment approach, which shows that Max $k$-Weight SAT can be essentially ``reduced'' to Max $k$-Coverage with a small overhead in the running time. While there seems to be some similarity between their approach and our Step I, it is unclear how to apply their technique directly to obtain an approximate kernel.

\section{Preliminaries}

For convenience, we represent any solution to Max $k$-Weight SAT as the set $Y \subseteq \cV$ (such that $|Y| \leq k$) of variables that are set to one. 
We write $\val_{\Phi}(Y)$ to denote the number of clauses in $\Phi$ satisfied by $Y$.
Let $\opt_{\Phi, k}$ denote the optimum, i.e. $\opt_{\Phi, k} = \max_{Y \in \binom{\cV}{\leq k}} \val_{\Phi}(Y)$.


In the subsequent analyses, it is useful to allow \emph{additive} errors in the approximation ratio too (rather than just multiplicative as in \cite{LPRS17}). Thus, let us also define $(\alpha, \gamma)$-APPA to be exactly the same as $\alpha$-APPA except that the output of $\cB$ is only required to be an $(\alpha\beta - \gamma)$-approximate solution. The following lemma allows us to relate this new notion to the standard one:

\begin{lemma} \label{lem:add-rem}
For any $\eps_1, \eps_2, c \in (0, 1)$, suppose that a maximization problem admits a polynomial-time $c$-approximation algorithm and an $(1 - \eps_1, \eps_2)$-APPA. Then it admits an $(1 - \eps_1 - \eps_2 / c)$-APPA with the same reduction algorithm. 
\end{lemma}

\begin{proof}
Let $(\cA, \cB)$ be the $(1 - \eps_1, \eps_2)$-APPA. We use the same reduction algorithm, but use the following solution lifting algorithm: output the best solution between one returned by $\cB$ and the approximation algorithm. It is simple to check that this is an $(1 - \eps_1 - \eps_2 / c)$-APPA as desired.
\end{proof}

We sometimes abuse the notation and refer to $\cA$ itself as the APPA/kernel and leave $\cB$ implicit.

We also recall the approximation algorithm for Max $k$-Weight SAT discussed in the introduction:
\begin{theorem}[\cite{Sviridenko01}] \label{lem:polytime-apx}
There is a polynomial-time $\left(1 - \frac{1}{e}\right)$-approximation for Max $k$-Weight SAT.
\end{theorem}

As mentioned earlier, our kernel often involves deleting variables or clauses. We abstract the conditions required for such pre-processing algorithms to be APPA in the two lemmas below. Here we use $\Iden$ to denote the identity solution lifting algorithm (i.e. one that outputs the input).

\begin{lemma}[Clause Modification APPA] \label{lem:clause-modif}
Suppose that $\cA$ is a parameter-preserving reduction algorithm for Max $k$-Weight SAT that also preserves the set of variables, i.e. on input $(\Phi = (\cV, \cC), k)$, it produces $(\Phi' = (\cV', \cC'), k)$ with $\cV' = \cV$. If there exists $\delta, h \geq 0$ and $s > 0$ (where $h,s$ can depend on $(\Phi, k)$) such that the following holds for all solution $Y$:
\begin{align} \label{eq:clause-preservation}
|\val_{\Phi}(Y) - s \cdot \val_{\Phi'}(Y) - h| \leq \delta \cdot \opt_{\Phi, k},
\end{align}
then $(\cA, \Iden)$ is an $(1, 2\delta)$-APPA.
\end{lemma}

\begin{proof}
Consider any $\beta$-approximate solution $Y$ to $(\Phi', k)$. Let $Y^*$ denote the optimum solution for $(\Phi, k)$. We can conclude that
\begin{align*}
\val_\Phi(Y) &\geq s \cdot \val_{\Phi'}(Y) + h - \delta \cdot \opt_{\Phi, k} \\
&\geq s\beta \cdot \val_{\Phi'}(Y^*) + h - \delta \cdot \opt_{\Phi, k} \\
&\geq s\beta \cdot \left(\frac{1}{s}\left(\val_{\Phi}(Y^*) - h - \delta \cdot \opt_{\Phi, k}\right)\right) + h - \delta \cdot \opt_{\Phi, k} 
&\geq \left(\beta - 2 \delta\right) \cdot \opt_{\Phi, k}. & \qedhere
\end{align*}
\end{proof}

\begin{lemma}[Variable Deletion APPA] \label{lem:var-deletion}
Suppose that $\cA$ is a parameter-preserving reduction algorithm for Max $k$-Weight SAT that just deletes a subset of variables (and all of their literals). If $\opt_{\Phi', k} \geq (1 - \delta) \cdot \opt_{\Phi, k}$ for some $\delta \in (0, 1)$,
then $(\cA, \Iden)$ is an $(1 - \delta)$-APPA.
\end{lemma}

\begin{proof}
Consider any $\beta$-approximate solution $Y$ to $(\Phi', k)$. Since $\Phi'$ results from deleting variables (and literals) from $\Phi$, we have $\val_\Phi(Y) \geq \val_{\Phi'}(Y) \geq \beta \cdot \opt_{\Phi', k} \geq \beta(1 - \delta) \cdot \opt_{\Phi, k}$.
\end{proof}

Finally, we use the following lemma which states a certain sparsity condition on $K_{a, b}$-free graphs. This lemma is very similar (but not exactly identical) to the classic K{\H{o}}v{\'a}ri-S{\'o}s-Tur{\'a}n bound~\cite{KST54} and also to \cite[Lemma 4.1]{JKPSSSU23}. We include the proof, which is almost the same as in those aforementioned work, in the appendix for completeness.

\begin{lemma}\label{lem:kst}
Let $a, b, n_L, n_R, d \in \N$ be such that $d \geq 2b, n_L \geq a \cdot \left(2n_R/d\right)^b$. Then, for any $K_{a, b}$-free bipartite graph with $n_L$ left vertices and $n_R$ right vertices, there exists a left vertex with degree $\leq d$. 
\end{lemma}

\section{Approximate Kernel for Max $k$-Weight SAT}

Throughout the remainder of this section, we let $a, b$ be any positive integers and $\eps$ be any real number in $(0, 1/4)$. For brevity, we will not state this assumption explicitly in the lemma statements.

For each variable $v$, let $\deg_{\Phi}(v)$ denote its degree in the incidence graph. 
For each clause $C$, let $\negg(C)$ denote the set of variables with negative literals in it, i.e. $\{v \in \cV \mid \neg v \in C\}$. Let $\cC_{\neg}$ denote the multiset of all clauses with at least one negative literal, i.e. $\{C \in \cC \mid \negg(C) \ne \emptyset\}$. 

\subsection{Step I: Reducing \# Negative Variables}


The first step of our reduction is described and analyzed below. Note that the condition that no clause contains $k+1$ negative literals is w.l.o.g. since these clauses are always true in any solution\footnote{See \Cref{sec:main-proof} for more detail.}.

\begin{lemma} \label{lem:ignore-positive-instances}
There is a  parameter-preserving $(1 - \eps)$-APPA for Max $k$-Weight SAT such that, if the input formula contains no clause with (at least) $k+1$ negative literals, then the output formula contains $O\left(\frac{k \log k}{\eps}\right)$ negative variables.
\end{lemma}

\begin{proof}
From \Cref{lem:add-rem} and \Cref{lem:polytime-apx}, it suffices to give an $(1, \eps)$-APPA with the claimed property.

For any instance $\Phi = (\cV, \cC)$ and a subset $\tcV \subseteq \cV$, let the \emph{normalized negative degree} of $v \in \tcV$ w.r.t. $\Phi, \tcV$ be defined as $\npdeg_{\Phi, \tcV}(v) := \sum_{C \in \cC \atop \negg(C) \ni v} \frac{1}{|\negg(C) \cap \tcV|}$. The reduction algorithm $\cA$ is iterative and, on input $(\Phi, k)$, it works as follows:
\begin{itemize}
\item Let $\tau = \frac{\eps}{2k} \cdot |\cC_{\neg}|$, and start with $\tcV_0 = \cV$ and $i = 0$.
\item While there exists $v_{i+1} \in \tcV_i$ such that $\npdeg_{\Phi, \tcV_i}(v_{i+1}) > \tau$, let $\tcV_{i+1} \gets \tcV_i \setminus \{v_{i+1}\}$ and increment $i$ by one.
\item Finally, output the formula $\Phi' = (\cV, \cC')$ where $\cC'$ results from removing all clauses containing a negative literal from $\tcV_i$. (Formally, $\cC' = \{C \in \cC \mid \negg(C) \cap \tcV_i = \emptyset\}$.)
\end{itemize}
Let $\ifin$ be the value of $i$ at the end of the algorithm. All negative variables in $\Phi'$ belong to $\cV \setminus \tcV_{\ifin}$. Thus, for the claimed number of negative variables in $\Phi'$, it suffices to show $\ifin = O\left(\frac{k \log k}{\eps}\right)$. By the while-loop, we have
\begin{align*}
\ifin \cdot \tau &< \sum_{i \in [\ifin]} \npdeg_{\Phi, \tcV_{i-1}}(v_i) 
&= \sum_{i \in [\ifin]} \sum_{C \in \cC \atop \negg(C) \ni v_i} \frac{1}{|\negg(C) \cap \tcV_{i-1}|} 
&= \sum_{C \in \cC_{\neg}} \sum_{i \in [\ifin] \atop v_i \in \negg(C)} \frac{1}{|\negg(C) \cap \tcV_{i-1}|}.
\end{align*}
Let us fix $C \in \cC_{\neg}$.
Notice that, for all $i$ such that $v_i \in \negg(C)$, $|\negg(C) \cap \tcV_{i-1}|$ are distinct because $v_i$ is removed from $\tcV_{i - 1}$ immediately after. Since we assume that $|\negg(C)| \leq k + 1$, we thus have
\begin{align*}
\ifin \cdot \tau 
< \sum_{C \in \cC_{\neg}} \left(\frac{1}{k + 1} + \frac{1}{k} + \cdots + 1\right)
\leq |\cC_{\neg}| \cdot \left(\ln(k+1) + 1\right).
\end{align*}
As a result, we have $\ifin \leq O\left(\frac{|\cC_{\neg}| \cdot \log k}{\tau}\right) = O\left(\frac{k \log k}{\eps}\right)$ as desired.

Let $\cCdel := \cC \setminus \cC'$ denote the multiset of deleted clauses.
We will next argue that \eqref{eq:clause-preservation} holds for $s = 1, h = |\cCdel|$ and $\delta = \eps/2$. 
To do this, consider any solution $Y \in \binom{\cV}{\leq k}$. First, it is obvious that
\begin{align} \label{eq:val-ub}
\val_{\Phi}(Y) \leq \val_{\Phi'}(Y) + |\cCdel|.
\end{align}
Next, let $\cCdel^{\UNSAT(Y)}$ denote the multiset of clauses in $\cCdel$ \emph{not} satisfied by $Y$. For $C \in \cCdel^{\UNSAT(Y)}$, we must have $\negg(C) \subseteq Y$. Recall that every $C \in \cCdel$ satisfies $\negg(C) \cap \tcV_{\ifin} \ne \emptyset$. This implies
\begin{align}
\left|\cCdel^{\UNSAT(Y)}\right| &= \sum_{C \in \cCdel^{\UNSAT(Y)}} \sum_{v \in \left(\negg(C) \cap \tcV_{\ifin}\right)} \frac{1}{|\negg(C) \cap \tcV_{\ifin}|} \nonumber \\
&\overset{(\spadesuit)}{=} \sum_{v \in (\tcV_{\ifin} \cap Y)} \sum_{C \in \cCdel^{\UNSAT(Y)} \atop \negg(C) \ni v} \frac{1}{|\negg(C) \cap \tcV_{\ifin}|} \nonumber \\
&\leq \sum_{v \in (\tcV_{\ifin} \cap Y)} \npdeg_{\Phi, \tcV_{\ifin}}(v) \nonumber \\
&\overset{(\heartsuit)}{\leq} |Y| \cdot \tau \nonumber \\
&\leq \frac{\eps}{2} \cdot |C_{\neg}|, \nonumber
\end{align}
where $(\spadesuit)$ follows from $\negg(C) \subseteq Y$ for all $C \in \cCdel^{\UNSAT(Y)}$ and $(\heartsuit)$ is from the while-loop condition. 

Next, observe that $\opt_{\Phi, k} \geq \val_{\Phi}(\emptyset) = |C_{\neg}|$. Combining this with the above, we then get
\begin{align} \label{eq:val-lb}
\val_{\Phi}(Y) = \val_{\Phi'}(Y) + \left|\cCdel\right| - \left|\cCdel^{\UNSAT(Y)}\right| \geq \val_{\Phi'}(Y) + \left|\cCdel\right| - 0.5\eps \cdot \opt_{\Phi, k}.
\end{align}
From \eqref{eq:val-ub} and \eqref{eq:val-lb}, we have that \eqref{eq:clause-preservation} holds for  $s = 1, h = |\cCdel|$ and $\delta = \eps/2$. Thus, \Cref{lem:clause-modif} implies that this is an $(1, \eps)$-APPA as desired.
\end{proof}

\subsection{Step II: Reducing \# Positive Variables}

\subsubsection{A Sunflower Lemma}

As mentioned earlier, this step will require a sunflower lemma-based reduction algorithm. We remark that the use of the sunflower lemma in kernelization is a standard technique; see e.g. \cite[Section 8]{kernelization-book}. In our application, we require a slightly better bound than the classic sunflower lemma~\cite{ER60}, which we will achieve under the $K_{a, b}$-free assumption. Below, we will state this lemma in terms of bipartite graphs instead of set systems, since this is more convenient for us.

We write $N_G(v)$ to denote the set of neighbors of $v$ in graph $G$; for a set of vertices $T$, we let $N_G(T) := \bigcup_{v \in T} N_G(v)$. In a bipartite graph $G = (A, B, E)$, a subset $S \subseteq A$ forms a \emph{sunflower} iff $N_G(v) \cap N_G(v')$ are the same for all distinct $v, v' \in S$. 
Our lemma is stated below:
\begin{lemma}[$K_{a,b}$-free Sunflower Lemma] \label{lem:sunflower}
For any $w, \ell \in \N$,
any $K_{a, b}$-free bipartite graph $G = (A, B, E)$ such that every vertex in $A$ has degree at most $\ell$ and $|A| \geq a((w-1)\ell)^b$ has a sunflower of size $w$. Moreover, such a sunflower can be found in polynomial time.
\end{lemma}

Compared to the standard bound (e.g. \cite{ER60}), the exponent here is $b$ instead of $\ell$. This improvement is crucial in our application below since we apply it for $\ell$ that is much larger than $b$.

\begin{proof}[Proof of \Cref{lem:sunflower}]
For convenience, we say that $G = (A, B, E)$ is $K_{a, 0}$-free if $|A| < a$. 

We prove the statement by induction on $b$. If $b = 0$, then this trivially holds by the above definition. Next, suppose that the statement holds for $b - 1$ for some $b \in \N$. To prove this statement for $b$, consider any $K_{a,b}$-free bipartite graph $G = (A, B, E)$ such that $|A| \geq a((w-1)\ell)^b$ and every vertex in $A$ has degree at most $\ell$. Consider any maximal set $T \subseteq A$ such that $N_G(v)$ are pairwise-disjoint for all $v \in T$. If $|T| \geq w$, then $T$ forms a sunflower of size (at least) $w$. Otherwise, if $|T| \leq w - 1$, then $|N_G(T)| \leq (w - 1)\ell$. Since $T$ is maximal, we have that $N_G(x) \cap N_G(T) \ne \emptyset$ for all $x \in A$. This means that there exists $u \in N_G(T)$ such that $N_G(u) \geq \frac{|A|}{(w - 1)\ell} \geq a((w-1)\ell)^{b - 1}$. Consider the subgraph of $G$ induced on $N_G(u) \cup (B \setminus \{u\})$. This is a $K_{a, b-1}$-free bipartite graph where $|N_G(u)| \geq a((w-1)\ell)^{b - 1}$. As such, we can apply the inductive hypothesis to conclude that there exists a sunflower $S \subseteq N_G(u)$ of size $w$ in this subgraph. Since $u$ is a common neighbor of all vertices in $S$ (w.r.t. $G$), $S$ is also a sunflower in $G$. This completes the inductive step. 

Note that this proof also yields a polynomial-time algorithm since computing a maximal set $T$ and finding $u$ can be done in polynomial time.
\end{proof}

\subsubsection{The Preprocessing Algorithm}

We next reduce the number of positive variables via a similar greedy-by-degree strategy to \cite{SkowronF17,Man19}. If the degrees of the vertices we select are all sufficiently large, then there is nothing else to be done (Case I below). However, if some vertex degrees are too small, we may need to keep other vertices (Case II below); we deal with this case using the sunflower lemma we showed above.

\begin{lemma} \label{lem:reduce-num-pos}
There is a parameter-preserving $(1 - \eps)$-APPA for $K_{a, b}$-free Max $k$-Weight SAT such that, if the input contains $\leq t$ negative variables, then the output has $\left(t + a \cdot O(b)^{2b} \cdot \frac{k^b}{\eps^{3b}}\right)$ variables.
\end{lemma}

\begin{proof}
Let $\cV_{\neg}$ denote the set of negative variables in the input formula $\Phi$.
Let $\cV_q$ be the set of $q$ positive variables with highest degrees for $q = k + a \cdot (2bk/\eps)^{b}$. Let $\tau$ denote the minimum degree of variables in $\cV_q$. We consider two cases based on the value of $\tau$.

\paragraph{Case I: $\tau \geq \frac{2 b}{\eps}$.}
In this case, we delete all variables outside of $\cV_{\neg} \cup \cV_q$ (and all their literals). Let $\Phi' = (\cV_{\neg} \cup \cV_q, \cC')$ denote the resulting formula. $\cA$ then outputs $(\Phi', k)$.

Below, we will argue that $\opt_{\Phi', k} \geq (1 - \eps) \cdot \opt_{\Phi, k}$. Note that this, together with \Cref{lem:var-deletion}, immediately implies that $(\cA, \Iden)$ is an $(1 - \eps)$-APPA.

To see that this is the case, let $Y^*$ denote the optimal solution in $\Phi$. Suppose that $Y^* \setminus (\cV_{\neg} \cup \cV_q) = \{u_1, \dots, u_p\}$. Consider the following iterative procedure:
\begin{itemize}
\item We start with $Y_0 \gets (Y^* \cap (\cV_{\neg} \cup \cV_q))$.
\item For $i = 1, \dots, p$:
\begin{itemize}
\item Pick $u^*_i = \argmax_{u \in \cV_q \setminus Y_{i-1}} \val_{\Phi}(Y_{i-1} \cup \{u\})$. (tie broken arbitrarily).
\item Let $Y_i \gets Y_{i - 1} \cup \{u^*_i\}$.
\end{itemize}
\end{itemize}
To compare $\val_{\Phi}(Y_p)$ and $\val_{\Phi}(Y^*)$, let us fix $i \in [p]$. First, since $u_i \notin \cV_{\neg} \cup \cV_q$, we have
\begin{align} \label{eq:opt-upper}
\val_\Phi(Y_0\cup \{u_1, \dots, u_i\}) - \val_\Phi(Y_0 \cup \{u_1, \dots, u_{i-1}\}) \leq \deg_{\Phi}(u_i) \leq \tau.
\end{align}
Next, let $\cC^{\SAT(Y_{i-1})}$ denote the multiset of clauses satisfied by $Y_{i - 1}$. Consider the subgraph of $\Gin$ induced on $\cV_q \setminus Y_{i - 1}$ on one side and $\cC^{\SAT(Y_{i-1})}$ on the other. By our choice of $q$, we have $\left|\cV_q \setminus Y_{i - 1}\right| \geq q - k \geq a \cdot (2k/\eps)^b$. Meanwhile, we also have  $|\cC^{\SAT(Y_{i-1})}| \leq \opt_{\Phi, k}$. Thus, we may apply \Cref{lem:kst} with $n_L = a \cdot (2k/\eps)^b, n_R = \opt_{\Phi, k}, d = \max\left\{\eps\tau, \frac{\eps}{k} \cdot \opt_{\Phi, k}\right\}$ to conclude that there exists $\tu_i \in \cV_q \setminus Y_{i-1}$ such that $|N_{\Gin_\Phi}(\tu_i) \cap \cC^{\SAT(Y_{i-1})}| \leq d$. This means that
\begin{align*}
\val_\Phi(Y_{i-1} \cup \{\tu_i\}) \geq \val_{\Phi}(Y_{i-1}) + \deg_{\Phi}(\tu_i) - d
\geq \val_{\Phi}(Y_{i-1}) + \tau - d,
\end{align*}
where the second inequality is from $\tu_i \in \cV_q$. Moreover, by our choice of $u^*_i$, we have 
\begin{align} \label{eq:constructed-lb}
\val_{\Phi}(Y_i) \geq \val_\Phi(Y_{i-1} \cup \{\tu_i\}) \geq \val_{\Phi}(Y_{i-1}) + \tau - d.
\end{align} 
By summing \Cref{eq:opt-upper} over all $i \in [p]$, we have
\begin{align*}
\opt_{\Phi, k} = \val_{\Phi}(Y^*) \leq \val_{\Phi}(Y_0) + p \cdot \tau.
\end{align*}
Moreover, by summing \Cref{eq:constructed-lb} over all $i \in [p]$ and then using the above inequality, we have
\begin{align*}
\val_{\Phi, k}(Y_p) &\geq \val_{\Phi, k}(Y_0) + p \cdot \tau - p \cdot d \\
&= \val_{\Phi, k}(Y_0) + p \cdot \tau - p \cdot \max\left\{\eps\tau, \frac{\eps}{k} \cdot \opt_{\Phi, k}\right\} \\
&\geq  \max\left\{\val_{\Phi, k}(Y_0) + (1 - \eps)p \cdot \tau, \left(1 - p \cdot \frac{\eps}{k}\right) \cdot \opt_{\Phi, k}\right\} \\
&\geq (1 - \eps) \cdot \opt_{\Phi, k}.
\end{align*}
This implies that $\opt_{\Phi', k} \geq \val_{\Phi, k}(Y_p) \geq (1 - \eps) \cdot \opt_{\Phi, k}$ as desired.

\paragraph{Case II: $\tau < \frac{2b}{\eps}$.} Let $\topt = \lceil \frac{k\tau}{\eps} \rceil$ and we instead use the following reduction algorithm:
\begin{itemize}
\item Start with the input formula $\Phi = (\cV, \cC)$.
\item Applying the following reduction rule until it cannot be applied:
\begin{itemize}
\item Let $\cV_{\deg \leq \tau}$ denote the set of positive vertices with degree at most $\tau$.
\item Use \Cref{lem:sunflower} on the subgraph of $\Gin$ induced on $\cV_{\deg \leq \tau} \cup \cC$.
\item If a sunflower of size $\topt + 1$ is found, then delete the variable with the lowest degree in the sunflower (tie broken arbitrarily) together with all its literals.
\end{itemize}
\item Let the final formula be $\Phi' = (\cV', \cC')$.
\end{itemize}
By \Cref{lem:sunflower}, there will be at most $a \cdot (\topt \cdot \tau)^b$ variables from $\cV_{\deg \leq \tau}$ left in $\cV'$. Thus, we have
\begin{align*}
|\cV'| 
\leq |\cV_{\neg}| + |\cV_{q}| + a \cdot (\topt \cdot \tau)^b &\leq t + a \cdot O(b)^{2b} \cdot \frac{k^b}{\eps^{3b}}.
\end{align*}

To show that this is an $(1 - \eps)$-APPA, we consider further two subcases:
\begin{itemize}
\item \textbf{Case II.A: } $\opt_{\Phi, k} > \topt$. In this case, let $Y^*$ denote the optimal solution in $\Phi$. We have
\begin{align*}
\opt_{\Phi', k} \geq 
\val_{\Phi}(Y^* \cap \cV') \geq \val_{\Phi}(Y^*) - \sum_{v \in (Y^* \setminus \cV')} \deg_{\Phi}(v) \geq \opt_{\Phi, k} - k \cdot \tau \geq (1 - \eps)\opt_{\Phi, k},
\end{align*}
where the third inequality follows from the fact that we only delete vertices with degree at most $\tau$. From the above inequality and \Cref{lem:var-deletion}, this is an $(1 - \eps)$-APPA as desired. 
\item \textbf{Case II.B: } $\opt_{\Phi, k} \leq \topt$. In this case, we argue that an application of the reduction rule does not change the optimum. To see this, suppose that we delete a vertex $v$ in a sunflower $T$ of size $\topt+1$. Either $v$ is not in the current optimal solution $Y^*$, or $v$ is in $Y^*$. In the former case, removing $v$ clearly does not change the optimum. In the latter case, since at most $\opt_{\Phi, k}$ clauses are satisfied and $|T| = \topt + 1 \geq \opt_{\Phi, k} + 1$, we can find another vertex $v' \in T \setminus \{v\}$ such that $N_G(v') \setminus N_G(v)$ does not have any clause that is satisfied by $Y^*$. As such, by replacing $v$ by $v'$ in $Y^*$, we have a solution with no less value than before. 

Thus, we have $\opt_{\Phi', k} = \opt_{\Phi, k}$ which, together with \Cref{lem:var-deletion}, implies that the reduction algorithm is also an $(1-\eps)$-APPA in this case. \qedhere
\end{itemize}
\end{proof}

\subsection{Step III: Reducing \# Clauses}

Finally, we reduce the number of clauses using a ``scaling and rounding of weights'' procedure, which is a standard technique in weighted-vs-unweighted reductions (see e.g.~\cite{CST01}).

\begin{lemma} \label{lem:clause-rep-dec}
There is a parameter-preserving $(1 - \eps)$-APPA for $K_{a, b}$-free Max $k$-Weight SAT such that the output formula has the same set of variables and $O\left(b \cdot (2n)^{a + 1} / \eps\right)$ clauses.
\end{lemma}

\begin{proof}
From \Cref{lem:add-rem} and \Cref{lem:polytime-apx}, it suffices to give an $(1, \eps)$-APPA with the claimed property.

On input $(\Phi, k)$, the reduction algorithm works as follows.
\begin{itemize}
\item Use \Cref{lem:polytime-apx} to compute $\topt$ s.t. $\opt_{\Phi, k} \geq \topt \geq \left(1 - \frac{1}{e}\right)\cdot \opt_{\Phi, k}$. Let $s := \frac{\eps \cdot \topt}{10 b \cdot (2n)^a}$. 
\item Let $\cCdist$ denote the set of distinct clauses in $\cC$.
\item Start with $\cC'$ being the empty multiset. For each $C \in \cCdist$, let $m_C$ denote the number of occurrences of $C$ in $\cC$ and add $\lfloor m_C / s \rfloor$ copies of $C$ to $\cC'$.
\item Output $(\Phi' = (\cV, \cC'), k)$.
\end{itemize}
To bound $|\cC'|$, notice that every variable $v \in \cV$ satisfies $\deg_{\Phi}(v) \leq 2 \opt_{\Phi, k}$; otherwise, setting $v$ to true or false alone would already satisfy more than $\opt_{\Phi, k}$ clauses. As a result, we have $|\cC| \leq 2n \cdot \opt_{\Phi, k}$. By our definition of $\cC'$, we thus have $|\cC'| \leq \frac{|\cC|}{s} \leq O(b \cdot (2n)^{a + 1} / \eps)$.

We claim that, for every solution $Y \in \binom{\cV}{\leq k}$, we have $\left|\val_{\Phi}(Y) - s \cdot \val_{\Phi'}(Y)\right| \leq \frac{\eps}{2} \cdot \opt_{\Phi, k}$. From this and \Cref{lem:clause-modif}, we can conclude that $(\cA, \Iden)$ forms an $(1, \eps)$-APPA as desired.

To see that the claim holds, note that
\begin{align*}
\left|\val_{\Phi}(Y) - s \cdot \val_{\Phi'}(Y)\right| \leq \sum_{C \in \cCdist} |m_C - s \cdot \lceil m_C / s\rceil|\leq s \cdot |\cCdist|.
\end{align*}
Now, since $\Phi$ is $K_{a, b}$-free, any set of $a$ variables can occur together in at most $b$ clauses. Thus, the number of clauses with at least $a$ literals is at most $b \cdot n^a$. Meanwhile, the number of \emph{unique} clauses with less than $a$ literals is at most $(2n)^{a}$. Plugging this into the above, we have
\begin{align*}
\left|\val_{\Phi}(Y) - s \cdot \val_{\Phi'}(Y)\right| \leq s \cdot (b \cdot n^a + (2n)^a) \leq \frac{\eps}{2} \cdot \opt_{\Phi, k},
\end{align*}
where the inequality is due to our choice of $s$. 
\end{proof}

\subsection{Putting Things Together: Proof of \Cref{thm:main-kernel}}
\label{sec:main-proof}

\begin{proof}[Proof of \Cref{thm:main-kernel}]
On input $(\Phi, k)$, the reduction algorithm works as follows:
\begin{enumerate}
\item Delete all clauses with at least $k + 1$ negative literals.
\item Apply $(1 - \eps/3)$-APPA reduction from \Cref{lem:ignore-positive-instances}.
\item Apply $(1 - \eps/3)$-APPA reduction from \Cref{lem:reduce-num-pos}.
\item Apply $(1 - \eps/3)$-APPA reduction from \Cref{lem:clause-rep-dec}.
\end{enumerate}
All clauses deleted in the first step are always true in any solution; therefore, \eqref{eq:clause-preservation} is satisfied $s = 1, h = $ \# deleted clauses and $\delta = 0$. Thus, by \Cref{lem:clause-modif} the first step (together with identity solution lift) is a 1-APPA. By our construction, the remaining steps are $(1 - \eps/3)$-APPA. Thus, the entire algorithm is an $(1 - \eps/3)^3 \geq (1 - \eps)$-APPA as desired.

As for the size, the second step ensures that there are $O\left(\frac{k \log k}{\eps}\right)$ negative variables left. The third step then ensures that the total number of variables is $n' = O\left(\frac{k \log k}{\eps}\right) + a \cdot O(b)^{2b} \cdot \frac{k^b}{\eps^{3b}}$. The last step then guarantees that the number of clauses is $O(b\cdot (2n')^{a+1}/\eps) \leq (k/\eps)^{O(ab)}$.
\end{proof}

\section{Discussion and Open Questions}
\label{sec:open}

In this work, we give an approximate kernel for Max $k$-Weight SAT based on (relatively) simple greedy strategies together with a sunflower lemma-based reduction rule. We remark that, although we assume that $\cC$ is a \emph{multiset}, we can also produce an instance for the \emph{set} version as follows: first, replicate each clause in the output instance $\lceil k/\eps \rceil$ times. Then, for every clause in the resulting instance, create a fresh new variable and add it to that clause. It is not hard to see that this reduction procedure is an $(1 - \eps)$-APPA and the final instance has no duplicated clauses. The size of the kernel remains $(k/\eps)^{O(ab)}$ after this transformation. 

Another interesting observation is that our APPA for reducing the number of negative variables (\Cref{lem:ignore-positive-instances}) does not require the $K_{a, b}$-free assumption on the incidence graph. Thus, it is applicable beyond the context of this work, e.g. for other graph classes or for restricted classes of CSPs.

A clear open question from our work is whether we can improve the size of the kernel further. In particular, is the exponent $b$ on the number of variables in \Cref{thm:main-kernel} necessary? Similarly, we can also ask whether the running time can be improved, although the gap here is smaller. Namely, can we remove the $\log k$ dependency in the exponent in \Cref{cor:alg}?

Another interesting direction is to consider other types of constraints beyond cardinality constraints. For example, Sellier~\cite{Sellier23} gave an approximate kernel and FPT-AS for Max $k$-Coverage with bounded frequency under matroid constraints, i.e. the solution $Y$ has to be an independent set of a given matroid. It is interesting whether we can relax the bounded frequency assumption to $K_{a, b}$-freeness similar to what~\cite{JKPSSSU23} and we have done for cardinality constraints.

\bibliographystyle{alpha}
\bibliography{ref}

\appendix

\section{Proof of \Cref{lem:kst}}

\begin{proof}[Proof of \Cref{lem:kst}]
Suppose for the sake of contradiction that there exists a $K_{a, b}$-free bipartite graph with $n_L$ left vertices and $n_R$ right vertices such that every vertex on the left has degree at least $d + 1$. The number of $K_{1, b}$ subgraph in this graph is at least $n_L \cdot \binom{d+1}{b}$. By pigeon-hole principle, this means that at least $\left\lceil \frac{n_L \cdot \binom{d+1}{b}}{\binom{n_R}{b}} \right\rceil$ such subgraphs shares the same set of $b$ vertices on the right. Meanwhile, from our assumptions on parameters, we have
\begin{align*}
\frac{n_L \cdot \binom{d+1}{b}}{\binom{n_R}{b}} \geq n_L \cdot \left(\frac{d+2-b}{n_R+1-b}\right)^b \geq n_L \cdot \left(\frac{d/2}{n_R}\right)^b \geq a,
\end{align*}
which contradicts with the assumption that the graph is $K_{a, b}$-free.
\end{proof}

\end{document}